\documentclass[conference]{IEEEtran}

\usepackage[mathscr]{euscript}
\usepackage{tikz,amsmath,amsthm,amssymb,dsfont}
\usetikzlibrary{shapes,arrows}

\newcommand{\mat}{\boldsymbol}
\newcommand{\abs}[1]{\left\lvert #1\right\rvert}
\newcommand{\abss}[1]{\lvert #1\rvert}
\newcommand{\tvs}[1]{\lVert #1\rVert_{\text{\tiny TV}}}
\newcommand{\E}[1]{\operatorname{\mathds{E}}\left\{#1\right\}}
\newcommand{\Es}[1]{\operatorname{\mathds{E}}\{#1\}}
\newcommand{\ind}[1]{\operatorname{\mathds{1}}\left\{#1\right\}}

\theoremstyle{plain}
\newtheorem{theoremflat}{Theorem}
\newenvironment{theorem}
  {\pushQED{\qed}\theoremflat}
  {\popQED\endtheoremflat}
  
\newtheorem{lemmaflat}[theoremflat]{Lemma}
\newenvironment{lemma}
  {\pushQED{\qed}\lemmaflat}
  {\popQED\endlemmaflat}  

\theoremstyle{definition}
\newtheorem{definitionflat}[theoremflat]{Definition}
\newenvironment{definition}
  {\pushQED{\qed}\definitionflat}
  {\popQED\enddefinitionflat}

\newtheorem{exampleflat}[theoremflat]{Example}
\newenvironment{example}
  {\pushQED{\qed}\exampleflat}
  {\popQED\endexampleflat}

\begin{document}
\title{Communication and Interference Coordination}

\author{\IEEEauthorblockN{Ricardo Blasco-Serrano, Ragnar Thobaben, and Mikael Skoglund}
\IEEEauthorblockA{KTH Royal Institute of Technology and ACCESS Linnaeus Centre\\
SE-100 44, Stockholm, Sweden\\
E-mail: \{ricardo.blasco, ragnar.thobaben, mikael.skoglund\}@ee.kth.se}}


\maketitle

\begin{abstract}
We study the problem of controlling the interference created to an external observer by a communication processes. We model the interference in terms of its type (empirical distribution), and we analyze the consequences of placing constraints on the admissible type. Considering a single interfering link, we characterize the communication-interference capacity region. Then, we look at a scenario where the interference is jointly created by two users allowed to coordinate their actions prior to transmission. In this case, the trade-off involves communication and interference as well as coordination. We establish an achievable communication-interference region and show that efficiency is significantly improved by coordination.
\end{abstract}

\IEEEpeerreviewmaketitle

\section{Introduction}
Communication is subject to undesirable and often unavoidable interference that degrades the performance of neighboring transceivers and impairs the operation of nearby electronic devices. From an information-theoretic point of view, interference has traditionally been studied using the interference channel, which models the mutual effects between two user pairs that communicate simultaneously. This channel abstraction captures the fundamental tradeoff between the communication rates of the two pairs. In spite of decades of efforts, our understanding of this tradeoff is only partial or restricted to some special cases (see \cite[Chapter~6]{GK11} for a basic summary). In addition, the model 
is less appropriate for the cases where the impairment is created to a different type of device that is not necessarily communicating. An alternative view of interference that goes beyond communication-impairment effects was proposed in \cite{BG11}. The authors modeled the communication-induced disturbances in terms of the undesired information rate and investigated the limits on the communication rate imposed by a constraint on the disturbance. They characterized explicitly the rate-disturbance region for the single disturbance case and gave partial results for other cases.

In this work, we take a similar approach although our model for the interference is quite different. Instead of endowing the interference with an informational meaning, we characterize it in terms of its type (i.e., empirical distribution). Thus, we study which communication rates are compatible with constraints placed on the type of the interference created by the communication process.  
Our results are therefore related to the study of channels with constraints on the channel inputs (e.g., see \cite[Sec.~3.3]{GK11} and references therein) and on the channel outputs \cite[Sec.~29]{Sha48}. Our motivation is similar to that in \cite{Gas07}, where output constraints were used as a model for the external power restrictions encountered, for example, in cognitive radio systems. As we shall see, our results for the single user can be interpreted as a generalization of those in \cite{Gas07} for discrete channels. Moreover, our work is also connected to \cite{SV97}, which studies the empirical distributions of capacity-achieving codes, although our codes are characterized both by communication properties (i.e., vanishing error probabilities) \emph{and} interference constraints (i.e., convergence of the interference type in an appropriate sense). 

We also consider a multiuser set-up in which the transmitters are allowed to coordinate their actions to mitigate the joint effect of their interference and improve the overall efficiency. This is closely related to the problem of coordination in networks, which was studied in \cite{CPC10}. Most relevant to our work, the authors characterized (empirical) coordination in terms of the type of the sequences of actions and established the fundamental limits for a variety of network topologies. We show that this framework for coordination is very useful when different transmitters are subject to a common interference constraint.

In the remainder of this section we introduce the basic mathematical concepts and establish the notation. We consider the single user case in Section~\ref{sec:single_user} and a multiple user case in Section~\ref{sec:multiple_user}. Finally, we conclude our work in Section~\ref{sec:concl}.


\subsection{Preliminaries}
We consider exclusively random variables with finite alphabets. We denote them and their realizations using upper case and lower case letters, respectively (e.g., $X$ and $x$). We use bold face for vectors and specify their lengths using superindices (e.g., $\mat{x}^n$). We use calligraphic letters (e.g., $\mathcal{T}$ or $\mathscr{T}$) to denote sets. Given a set $\mathcal{T}$, we denote its complement by $\mathcal{T}^c$.

\begin{definition}[Total Variation]\label{def:tv}
Let $P_{X,Y}$ and $Q_{X,Y}$ be two probability distributions defined on $\mathcal{X}\times\mathcal{Y}$. The total variation between them is defined as
\begin{align*}
\tvs{P_{X,Y}-Q_{X,Y}}\triangleq\frac{1}{2}\sum_{x,y}\abs{P_{X,Y}(x,y)-Q_{X,Y}(x,y)}.
\end{align*}
\end{definition}

\begin{definition}[Type]\label{def:type}
Let $\mat{x}^n\in\mathcal{X}^n$ and $\mat{y}^n\in\mathcal{Y}^n$. The type of the tuple $(\mat{x}^n,\mat{y}^n)$ is defined as
\begin{align*}
T_{\mat{x}^n,\mat{y}^n}(x,y)\triangleq\frac{1}{n}\sum_{i=1}^n\ind{(x_i,y_i)=(x,y)}
\end{align*}
for all $(x,y)\in\mathcal{X}\times\mathcal{Y}$, where $\ind{\cdot}$ is the indicator function. 
\end{definition}

\begin{definition}[Typical sequence]\label{def:typ_seq}
Let $\mat{x}^n\in\mathcal{X}^n$ and $\epsilon>0$. We say that the sequence $\mat{x}^n$ is ($\epsilon$-)typical with respect to a distribution $P_X$ if $\tvs{T_{\mat{x}^n}-P_X}<\epsilon$. We denote by $\mathcal{T}_{\epsilon}^{(n)}(P_X)$ the set of all such sequences.
\end{definition}

Most of our results involve the following notion of convergence of sequences of probability distributions. Consider a sequence (indexed by $n$) of random vectors $\mat{X}^n$ with $\mat{X}^n\sim P_{\mat{X}^n}$  for some sequence of distributions $P_{\mat{X}^n}$, and the corresponding sequence of types $T_{\mat{X}^n}$. 
Consider also a sequence of deterministic distributions $G^{(n)}$. We say that $T_{\mat{X}^n}$ converges in probability in total variation to $G^{(n)}$ if
\begin{align*}
\lim_{n\rightarrow\infty}\Pr(\tvs{T_{\mat{X}^n}-G^{(n)}}\ge\epsilon)=0
\end{align*}
for all $\epsilon>0$. We denote this using the shorthand notation 
\begin{align*}
\tvs{T_{\mat{X}^n}-G^{(n)}}\rightarrow 0\text{ in probability}.
\end{align*}
(The specialization of this notion of convergence to the case of fixed $G$ or to deterministic sequences is straightforward.)

\section{Single User}
\label{sec:single_user}

Consider the scenario depicted in Figure~\ref{fig:SingleUserScenario}. This corresponds to a discrete memoryless channel (DMC) with one input $X$ and two outputs $Y$ and $Z$. The output $Y$ is the observation at the intended receiver, while $Z$ corresponds to an undesired interference created to an external observer. The channel is governed by a conditional probability mass function (pmf) $P_{Y,Z|X}$. The encoder-decoder pair can use the channel for communicating a random message $M$ as long as the interference $\mat{z}^n$ has a certain shape, measured in terms of its type $T_{\mat{z}^n}(z)$. For this purpose, they use a code.

\begin{definition}[Code]
An $(n,2^{nR})$-code for the scenario in Figure~\ref{fig:SingleUserScenario} consists of:
\begin{itemize}
\item a message set $\mathcal{M}\triangleq\{1,\ldots,\lceil2^{nR}\rceil\}$,
\item an encoding function $\mat{x}^n:\mathcal{M}\rightarrow\mathcal{X}^n$,
\item a decoding function $\hat{m}:\mathcal{Y}^n\rightarrow\mathcal{M}\cup\{e\}$.
\end{itemize}
\end{definition}
We assume that the message is uniformly distributed over the message set.

\begin{definition}[Achievability]
We say that the communication rate $R$ is achievable with interference type $G_Z$ if there exists a sequence of $(n,2^{nR})$-codes such that
\begin{align}
\lim_{n\rightarrow\infty}\Pr(\hat{M}\ne M)&=0,\label{eq:achiev_cond1}\\
\tvs{T_{\mat{Z}^n}-G_Z}&\rightarrow 0\text{~~in probability}\label{eq:achiev_cond2}
\end{align}
under the distribution induced by the codes.
\end{definition}

The \emph{communication-interference capacity region} $\mathcal{C}$ of the DMC $P_{Y,Z|X}$ is the closure of the set of all rate-interference type tuples $(R,G_Z)$ that are achievable. 

\begin{figure}[!t]
\centering
\resizebox{\columnwidth}{!}{
\tikzstyle{line} = [draw, ->]
\tikzstyle{block} = [rectangle, draw, text width=5em, text centered, rounded corners, minimum height=3em]
\tikzstyle{channel} = [rectangle, draw, text width=5em, text centered, minimum height=3em]
\tikzstyle{invisible} = [draw=none]
\begin{tikzpicture}[node distance = 4cm, auto]
    \node [block] (Encoder_1) {Encoder};
	\node [channel, right of=Encoder_1, minimum height=3cm] (ch_1) {$P_{Y,Z|X}$};
    \node [block, right of=ch_1, yshift=1cm] (Decoder_1) {Decoder};
       
    \path [line] (Encoder_1) -- ([xshift=.5cm]Encoder_1.east)[above] node{$\mat{X}^n$} -- (ch_1) ;
    \path [line] ([yshift=1cm]ch_1.east) -- ([yshift=1cm,xshift=.5cm]ch_1.east) [above right] node{$\mat{Y}^n$} -- (Decoder_1.west) ;
	\path [line] ([yshift=-1cm]ch_1.east) -- ([yshift=-1cm,xshift=.5cm]ch_1.east) [right] node{$\mat{Z}^n$};	
	\path [line] ([xshift=-.5cm]Encoder_1.west) [left] node{$M$} -- (Encoder_1.west);
	\path [line] (Decoder_1.east) -- ([xshift=.5cm]Decoder_1.east) [right] node{$\hat{M}$};		
\end{tikzpicture}
}
\caption{Scenario for single-user communication with interference constraint.}
\label{fig:SingleUserScenario}
\end{figure}
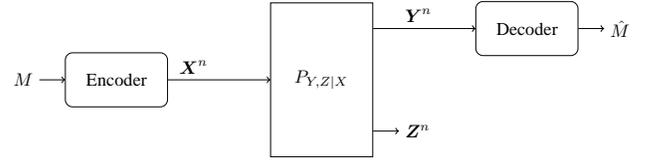

Our main result for the channel model in Figure~\ref{fig:SingleUserScenario} is a complete characterization of the communication-interference capacity region (Theorem~\ref{th:capacity_interference_pmf}). This region is convex and depends only on the marginals $P_{Y|X}$ and $P_{Z|X}$. Convexity is easily proven using standard time-sharing arguments. The dependency on the marginals also follows from well-known arguments (see e.g., \cite[Lemma~5.1]{GK11}).

\begin{theorem}\label{th:capacity_interference_pmf}
The communication-interference capacity region $\mathcal{C}$ of the DMC $P_{Y,Z|X}$ is the set of rate-interference type tuples $(R,G_Z)$ such that
\begin{align*}
R\le \max_{P_X\in\mathcal{P}}I(X;Y)
\end{align*}
where
\begin{align}
\mathcal{P}\triangleq\left\{ P_X:\sum_x P_XP_{Z|X}=G_Z\right\}.\label{eq:set_P}
\end{align}
\end{theorem}

Observe that this result agrees with our basic understanding of communication and coordination. In particular, the capacity expression is reminiscent of that for the point-to-point channel but the maximization is over the restricted set $\mathcal{P}$ of input distributions $P_X$ that induce the desired interference type $G_Z$. We will refer to the set $\mathcal{P}$ defined in \eqref{eq:set_P} as the \emph{pre-image} of $G_Z$. It is simple to show that the pre-image of a given $G_Z$ is a closed and convex set.

The result in Theorem~\ref{th:capacity_interference_pmf} is different from those involving constraints on the channel output in \cite[Sec.~29]{Sha48} and \cite{Gas07}. For example, satisfying an interference power constraint does not directly imply convergence of the type of the interference in the sense defined above. In contrast, convergence of the type ensures that the power constraint is satisfied. However, our characterization of the interference in terms of its type does not extend to continuous alphabets. 


In the remainder of this section we will prove Theorem~\ref{th:capacity_interference_pmf}. For this purpose, we first introduce the following auxiliary results (Lemmas~\ref{lemma:p_not_empty}-\ref{lemma:conv_output_input}).
%

\begin{lemma}\label{lemma:p_not_empty}
The interference type $T_{\mat{Z}^n}$ induced by a sequence of $(n,2^{nR})$-codes can only converge in probability to distributions $G_Z$ with non-empty pre-image, that is, $\mathcal{P}\ne\emptyset$.
\end{lemma}

\begin{proof}
First, observe that convergence in probability
\begin{align*}
\tvs{T_{\mat{Z}^n}-G_Z}\rightarrow0
\end{align*}
implies that 
\begin{align*}
\Es{\tvs{T_{\mat{Z}^n}-G_Z}}\rightarrow0
\end{align*}
because the total variation is bounded. In turn, this means that
\begin{align*}
\Es{T_{\mat{Z}^n}}\rightarrow G_Z
\end{align*} 
by a simple application of Jensen's inequality. Now, note that
\begin{align*}
\Es{T_{\mat{Z}^n}}&=\sum_{x}\Es{T_{\mat{X}^n,\mat{Z}^n}}\\
&=\sum_{x}\Es{T_{\mat{X}^n}}P_{Z|X}\\
&=f(\Es{T_{\mat{X}^n}}),
\end{align*}
where $f:\mathcal{X}\rightarrow\mathcal{Z}$ is a continuous function and $\Es{T_{\mat{X}^n}}$ is a bounded sequence of probability distributions on $\mathcal{X}$. Thus, by the Bolzano-Weierstrass theorem \cite[Theorem~3.6]{Rud76}, the sequence $\Es{T_{\mat{X}^n}}$ has a convergent subsequence, which we denote by $\bar{P}^{(n)}_X$. That is,
\begin{align*}
\bar{P}^{(n)}_X\rightarrow \hat{P}_X,
\end{align*}
where $\hat{P}_X$ is the corresponding limit (i.e., a probability distribution on $\mathcal{X}$). By convergence $\Es{T_{\mat{Z}^n}}\rightarrow G_Z$ and by continuity of the function $f$, we establish that
\begin{align*}
\lim_{n\rightarrow\infty}f(\Es{T_{\mat{X}^n}})&=\lim_{n\rightarrow\infty}f(\bar{P}^{(n)}_X)\\
&=f(\hat{P}_X)\\
&=G_Z.
\end{align*}
This means that $\hat{P}_X(x)\in\mathcal{P}$. Therefore, $\mathcal{P}\ne\emptyset$.
\end{proof}

\begin{lemma}\label{lemma:positive_distance2}
Let $G_Z$ be given and have pre-image $\mathcal{P}$ such that $\mathcal{P}\ne\emptyset$ and $\mathcal{P}^c\ne\emptyset$.
Consider the sets
\begin{align*}
\tilde{\mathcal{P}}_\epsilon&\triangleq\{\tilde{P}_X: \tvs{\tilde{P}_X-P_X}\ge\epsilon\text{ for all }P_X\in\mathcal{P}\},\\
\tilde{\mathcal{G}}_\epsilon&\triangleq\left\{\tilde{G}_Z: \sum_x P_{Z|X}\tilde{P}_X=\tilde{G}_Z\text{ for some }\tilde{P}_X\in\tilde{\mathcal{P}}_\epsilon\right\},
\end{align*}
defined for any fixed $\epsilon>0$ such that $\tilde{\mathcal{P}}_\epsilon\ne\emptyset$. Let
\begin{align*}
d^\star=\inf_{\tilde{G}_Z\in\tilde{\mathcal{G}}_\epsilon}\tvs{G_Z-\tilde{G}_Z}.
\end{align*}
Then, we have that $d^\star>0$.
\end{lemma}

\begin{proof}
Assume that $d^\star=0$. Note that $\tilde{\mathcal{P}}_\epsilon$ is a compact set and that $\tilde{G}_Z$ is a continuous function of $\tilde{P}_X$. Therefore,  $\tilde{\mathcal{G}}_\epsilon$ is a compact set, too. Note also that $\tvs{G_Z-\tilde{G}_Z}$ is a continuous function of $\tilde{G}_Z$. Thus, by Weierstrass' extreme value theorem \cite[Theorem~4.16]{Rud76}, there must exist some $\tilde{G}_Z\in\tilde{\mathcal{G}}_\epsilon$ (and hence some $\tilde{P}_X\in\tilde{\mathcal{P}}_\epsilon$) such that
\begin{align*}
\tvs{G_Z-\tilde{G}_Z}=0.
\end{align*}
That is, $G_Z=\tilde{G}_Z$. However, this would imply that $\tilde{P}_X\in\mathcal{P}$, which is a contradiction. Thus, we must have $d^\star>0$.
\end{proof}

\begin{lemma}\label{lemma:disjoint_sets}
Let $\epsilon>0$ and consider two arbitrary pmfs $Q_Z$ and $\tilde{Q}_Z$ defined on $\mathcal{Z}$ with typical sets $\mathcal{T}_{\epsilon}^{(n)}(Q_Z)$ and $\mathcal{T}_{\epsilon}^{(n)}(\tilde{Q}_Z)$, respectively. If the total variation between the pmfs satisfies $\tvs{Q_Z-\tilde{Q}_Z}>2\epsilon$ then the two typical sets are disjoint. That is, $\mathcal{T}_{\epsilon}^{(n)}(Q_Z)\cap\mathcal{T}_{\epsilon}^{(n)}(\tilde{Q}_Z)=\emptyset$.
\end{lemma}

\begin{proof}
Let $\mat{z}^n\in\mathcal{T}_{\epsilon}^{(n)}(Q_Z)$, that is,
\begin{align*}
\tvs{Q_Z-T_{\mat{z}^n}}<\epsilon.
\end{align*}
Then
\begin{align*}
\tvs{\tilde{Q}_Z-T_{\mat{z}^n}}&=\tvs{\tilde{Q}_Z-Q_Z+Q_Z-T_{\mat{z}^n}}\\
&\ge\tvs{\tilde{Q}_Z-Q_Z}-\tvs{Q_Z-T_{\mat{z}^n}}\\
&>2\epsilon-\epsilon.
\end{align*}
Thus $\mat{z}^n\notin\mathcal{T}_{\epsilon}^{(n)}(\tilde{Q}_Z)$ and $\mathcal{T}_{\epsilon}^{(n)}(Q_Z)\cap\mathcal{T}_{\epsilon}^{(n)}(\tilde{Q}_Z)=\emptyset$.
\end{proof}

\begin{lemma}\label{lemma:conv_output_input}
Let $G_Z$ be fixed and have pre-image $\mathcal{P}$. If a sequence of $(n,2^{nR})$-codes induces an interference type $T_{\mat{Z	}^n}$ such that 
\begin{align}
\tvs{T_{\mat{Z}^n}-G_Z}&\rightarrow 0\text{~~in probability,}\label{eq:conv_output_type}
\end{align}
then the expectation of the type of the codewords $\E{T_{\mat{X}^n}}$ satisfies
\begin{align}
\tvs{\E{T_{\mat{X}^n}}-P^{(n)}_X}&\rightarrow 0\label{eq:conv_input_etype}
\end{align}
for some sequence $P_X^{(n)}$ with $P_X^{(n)}\in\mathcal{P}$ for all $n$.
\end{lemma}

\begin{proof}
First, note that $\mathcal{P}\ne\emptyset$ by virtue of Lemma~\ref{lemma:p_not_empty}. Moreover, if $\mathcal{P}$ is equal to the whole simplex of probability distributions on $\mathcal{X}$ (i.e., $\mathcal{P}^c=\emptyset$) the proof is trivial. We prove the lemma for the case $\mathcal{P}\ne\emptyset,\mathcal{P}^c\ne\emptyset$ in two steps. i) First, we show that \eqref{eq:conv_output_type} implies that $\lim_{n\rightarrow\infty}\Pr(\mat{X}^n\notin\mathscr{T}_\epsilon^{(n)}(\mathcal{P}))=0$ for any $\epsilon>0$, where
\begin{align*}
\mathscr{T}_{\epsilon}^{(n)}(\mathcal{P})\triangleq\{\mat{x}^n: \tvs{T_{\mat{x}^n}-P_X}<\epsilon\text{ for some }P_X\in\mathcal{P}\}.
\end{align*}
(The set $\mathscr{T}_{\epsilon}^{(n)}$ is a straightforward generalization of the typical set $\mathcal{T}_{\epsilon}^{(n)}$.) ii) Then, we show that this implies \eqref{eq:conv_input_etype}. 

i) We prove the first step by contradiction. Assume that \eqref{eq:conv_output_type} is satisfied by some sequence of $(n,2^{nR})$-codes with distribution $P_{\mat{X}^n}$ for which there exist $\delta>0$ and $\epsilon_x>0$ such that 
\begin{align*}
\delta\le\limsup_{n\rightarrow\infty}\Pr(\mat{X}^n\notin\mathscr{T}_{\epsilon_x}^{(n)}(\mathcal{P})).
\end{align*}
Note that for every $\epsilon_x'$ such that $0<\epsilon_x'<\epsilon_x$ we have that $\tilde{P}_{\epsilon_x}\subseteq\tilde{P}_{\epsilon_x'}$ and this implies that
$\Pr(\mat{X}^n\notin\mathscr{T}_{\epsilon_x}^{(n)}(\mathcal{P}))\le\Pr(\mat{X}^n\notin\mathscr{T}_{\epsilon_x'}^{(n)}(\mathcal{P}))$.
For our purposes, it will be more convenient to write our expressions in terms of 
\begin{align*}
\tilde{\mathcal{P}}_{\epsilon_x}&\triangleq\{\tilde{P}_X:\tvs{\tilde{P}_X-P_X}\ge\epsilon_x\text{ for all }P_X\in\mathcal{P}\}.
\end{align*}
With this notation, the set $\{\mat{x}^n\notin\mathscr{T}_{\epsilon_x}^{(n)}(\mathcal{P})\}$ is equivalent to $\{\mat{x}^n:T_{\mat{x}^n}\in\tilde{\mathcal{P}}_{\epsilon_x}\}$. Observe that $\tilde{\mathcal{P}}_{\epsilon_x}\ne\emptyset$ for sufficiently small $\epsilon_x$ because $\tilde{\mathcal{P}}_{\epsilon_x}\subseteq\mathcal{P}^c$ and $\mathcal{P}^c$ is a set with non-empty interior.
Thus, without loss of generality, we assume that $\tilde{\mathcal{P}}_{\epsilon_x}\ne\emptyset$.

Now, we define the following finite cover $\mathcal{Q}_{\epsilon_c}$ of the set $\tilde{\mathcal{P}}_{\epsilon_x}$.  Given $\epsilon_c$ such that $0<\epsilon_c<\epsilon_x$, the set $\mathcal{Q}_{\epsilon_c}$ is a \emph{finite} set of distributions on $\mathcal{X}$ such that for every $\tilde{P}_X\in\tilde{\mathcal{P}}_{\epsilon_x}$ there exists some $P_X\in\mathcal{Q}_{\epsilon_c}$ with
\begin{align*}
\tvs{P_X-\tilde{P}_X}<\epsilon_c.
\end{align*}
Such a cover exists because the set $\tilde{\mathcal{P}}_{\epsilon_x}$ is compact. In fact, there exist more than one set with these properties. For convenience, we choose one (any) such set with the smallest possible cardinality. Thus, any distribution in $\tilde{\mathcal{P}}_{\epsilon_x}$ can be approximated by an element in the finite set $\mathcal{Q}_{\epsilon_c}$ with an error in terms of the total variation not exceeding $\epsilon_c$. Fix an arbitrary ordering of the elements in $\mathcal{Q}_{\epsilon_c}$
\begin{align*}
\mathcal{Q}_{\epsilon_c}=\{Q_{X,1},Q_{X,2},\ldots Q_{X,\abss{\mathcal{Q}_{\epsilon_c}}}\},
\end{align*}
and let 
\begin{align*}
\tilde{\mathcal{Q}}_{i}\triangleq\{\tilde{P}_X\in\tilde{\mathcal{P}}_{\epsilon_x}:\tvs{Q_{X,i}-\tilde{P}_X}<\epsilon_c\}
\end{align*}
for $i\in\{1,\ldots,\abs{\mathcal{Q}_{\epsilon_c}}\}$. To avoid the possibility that $\tilde{P}_X\in\tilde{\mathcal{Q}}_{i}$ and $\tilde{P}_X\in\tilde{\mathcal{Q}}_{j}$ for $i\ne j$, we define the following disjoint sets
\begin{align*}
\mathcal{Q}_1&\triangleq\tilde{\mathcal{Q}}_{1},\\
\mathcal{Q}_i&\triangleq\tilde{\mathcal{Q}}_{i}\backslash\bigcup_{j=1}^{i-1}\tilde{\mathcal{Q}}_{j}
\end{align*}
for $i\in\{2,\ldots,\abs{\mathcal{Q}_{\epsilon_c}}\}$. Observe that $\cup_i \mathcal{Q}_{i}=\tilde{\mathcal{P}}_{\epsilon_x}$.
Thus, for each $\mat{x}^n\notin\mathscr{T}_{\epsilon_x}^{(n)}(\mathcal{P})$ its type $T_{\mat{x}^n}$ satisfies $T_{\mat{x}^n}\in\mathcal{Q}_{i}$ for exactly one $i\in\mathcal\{1,\ldots,\abss{\mathcal{Q}_{\epsilon_c}}\}$. Using this covering into disjoints sets, we write
\begin{align*}
\sum_{\mat{x}^n\notin\mathscr{T}_{\epsilon_x}^{(n)}(\mathcal{P})}P_{\mat{X}^n}(\mat{x}^n)
&=
\sum_{i=1}^{\abs{\mathcal{Q}_{\epsilon_c}}} \sum_{\mat{x}^n:T_{\mat{x}^n}\in{\mathcal{Q}_i}}P_{\mat{X}^n}(\mat{x}^n).
\end{align*}

Now, for arbitrary $\epsilon>0$, write
\begin{align}
&\sum_{\mat{z}^n\notin\mathcal{T}_{\epsilon}^{(n)}(G_Z)}\hspace{-.5cm}P_{\mat{Z}^n}(\mat{z}^n)=\sum_{\mat{x}^n}P_{\mat{X}^n}(\mat{x}^n)\hspace{-.5cm}\sum_{\mat{z}^n\notin\mathcal{T}_{\epsilon}^{(n)}(G_Z)}\hspace{-.5cm}P_{\mat{Z}^n|\mat{X}^n}(\mat{z}^n|\mat{x}^n)\nonumber\\
&\ge\sum_{\mat{x}^n\notin\mathscr{T}_{\epsilon_x}^{(n)}(\mathcal{P})}P_{\mat{X}^n}(\mat{x}^n)\sum_{\mat{z}^n\notin\mathcal{T}_{\epsilon}^{(n)}(G_Z)}P_{\mat{Z}^n|\mat{X}^n}(\mat{z}^n|\mat{x}^n)\nonumber\\
&= 
\sum_{\mat{x}^n:T_{\mat{x}^n}\in{\mathcal{Q}_1}}P_{\mat{X}^n}(\mat{x}^n)\sum_{\mat{z}^n\notin\mathcal{T}_{\epsilon}^{(n)}(G_Z)}P_{\mat{Z}^n|\mat{X}^n}(\mat{z}^n|\mat{x}^n)\nonumber\\
&\quad\quad+\sum_{\mat{x}^n:T_{\mat{x}^n}\in{\mathcal{Q}_2}}P_{\mat{X}^n}(\mat{x}^n)\sum_{\mat{z}^n\notin\mathcal{T}_{\epsilon}^{(n)}(G_Z)}P_{\mat{Z}^n|\mat{X}^n}(\mat{z}^n|\mat{x}^n)\nonumber\\
&\quad\quad+~\ldots\label{eq:pzn_in_sums}
\end{align}
Consider the $i^{th}$ term in \eqref{eq:pzn_in_sums}. First, note that each of the sequences $\mat{x}^n$ in the sum belongs to the typical set $\mathcal{T}_{\epsilon_c}^{(n)}(Q_{X,i})$. Now, define $Q_{Z,i}\triangleq\sum_xP_{Z|X}Q_{X,i}$ and consider the set $\mathcal{T}_{\epsilon}^{(n)}(Q_{Z,i})$ of sequences $\mat{z}^n$ that are typical according to $Q_{Z,i}$.

From Lemma~\ref{lemma:positive_distance2} we know that, given $\epsilon_x$, there exists a fixed $d^\star>0$ such that $\tvs{G_Z-Q_{Z,i}}\ge d^\star$ for all $Q_{Z,i}$ ($i\in\{1,\ldots,\abs{\mathcal{Q}_{\epsilon_c}}\}$). Thus, for any $\epsilon$ such that $0<\epsilon<\frac{d^\star}{2}$, applying Lemma~\ref{lemma:disjoint_sets} we see that $\mathcal{T}_{\epsilon}^{(n)}(G_{Z})\cap\mathcal{T}_{\epsilon}^{(n)}(Q_{Z,i})=\emptyset$. Using this, we write
\begin{align*}
&\sum_{\mat{x}^n:T_{\mat{x}^n}\in{\mathcal{Q}_i}}P_{\mat{X}^n}(\mat{x}^n)\sum_{\mat{z}^n\notin\mathcal{T}_{\epsilon}^{(n)}(G_Z)}P_{\mat{Z}^n|\mat{X}^n}(\mat{z}^n|\mat{x}^n)\\
&\quad\quad\ge
\sum_{\mat{x}^n:T_{\mat{x}^n}\in{\mathcal{Q}_i}}P_{\mat{X}^n}(\mat{x}^n)\sum_{\mat{z}^n\in\mathcal{T}_{\epsilon}^{(n)}(Q_{Z,i})}P_{\mat{Z}^n|\mat{X}^n}(\mat{z}^n|\mat{x}^n).
\end{align*}
Moreover, by the conditional typicality lemma \cite[Lemma~2.12]{CK81}, we know that
\begin{align*}
\sum_{\mat{z}^n\in\mathcal{T}_{\epsilon}^{(n)}(Q_{Z,i})}P_{\mat{Z}^n|\mat{X}^n}(\mat{z}^n|\mat{x}^n)\ge 1-\delta_{\epsilon_c,\epsilon}(n)
\end{align*}
for every $\mat{x}^n$ such that $T_{\mat{x}^n}\in{\mathcal{Q}_i}$ and where $\delta_{\epsilon_c,\epsilon}(n)\triangleq\frac{1}{4n}\left(\frac{\abs{\mathcal{X}}\abs{\mathcal{Z}}}{\epsilon-\epsilon_c}\right)^2$. The term $\delta_{\epsilon_c,\epsilon}(n)$ goes to $0$ with $n$ and is fixed given the cover $\mathcal{Q}_{\epsilon_c}$. Thus,
\begin{align*}
\sum_{\mat{x}^n:T_{\mat{x}^n}\in{\mathcal{Q}_i}}P_{\mat{X}^n}(\mat{x}^n)\sum_{\mat{z}^n\notin\mathcal{T}_{\epsilon}^{(n)}(G_Z)}P_{\mat{Z}^n|\mat{X}^n}(\mat{z}^n|\mat{x}^n)&\\
\ge (1-\delta_{\epsilon_c,\epsilon}(n))\sum_{\mat{x}^n:T_{\mat{x}^n}\in{\mathcal{Q}_i}}P_{\mat{X}^n}(\mat{x}^n)&.
\end{align*}
Using this, we rewrite \eqref{eq:pzn_in_sums} as
\begin{align*}
\sum_{\mat{z}^n\notin\mathcal{T}_{\epsilon}^{(n)}(G_Z)}\hspace{-.25cm}P_{\mat{Z}^n}(\mat{z}^n)
&\ge \sum_{i=1}^{\abs{\mathcal{Q}_{\epsilon_c}}} \sum_{\mat{x}^n:T_{\mat{x}^n}\in{\mathcal{Q}_i}}\hspace{-.25cm}P_{\mat{X}^n}(\mat{x}^n)(1-\delta_{\epsilon_c,\epsilon}(n))\\
&\ge(1-\delta_{\epsilon_c,\epsilon}(n))\sum_{\mat{x}^n\notin\mathscr{T}_{\epsilon_x}^{(n)}(\mathcal{P})}P_{\mat{X}^n}(\mat{x}^n).
\end{align*}
Therefore, for any $0<\epsilon<\frac{d^\star}{2}$ we have
\begin{align*}
\limsup_{n\rightarrow\infty}\hspace{-.5cm}\sum_{\mat{z}^n\notin\mathcal{T}_{\epsilon}^{(n)}(G_Z)}\hspace{-.5cm}P_{\mat{Z}^n}(\mat{z}^n)
&\ge \limsup_{n\rightarrow\infty}(1-\delta_{\epsilon_c,\epsilon}(n))\hspace{-.5cm}\sum_{\mat{x}^n\notin\mathscr{T}_{\epsilon_x}^{(n)}(\mathcal{P})}\hspace{-.5cm}P_{\mat{X}^n}(\mat{x}^n)\\
&\ge\delta\\
&>0.
\end{align*}
This contradicts our initial hypothesis that $P_{\mat{X}^n}$ induces a type $T_{\mat{Z}^n}$ that satisfies \eqref{eq:conv_output_type}. Thus, we must have $\lim_{n\rightarrow\infty}\Pr(\mat{X}^n\notin\mathscr{T}^{(n)}_\epsilon(\mathcal{P}))=0$ for any $\epsilon>0$.

ii) Now, we show that this implies \eqref{eq:conv_input_etype}. To this end, we write
\begin{align}
&\tvs{\E{T_{\mat{X}^n}}-P^{(n)}_X}\nonumber\\
&=\lVert\Es{T_{\mat{X}^n}|\mat{X}^n\in\mathscr{T}_{\epsilon}^{(n)}(\mathcal{P})}\Pr(\mat{X}^n\in\mathscr{T}_{\epsilon}^{(n)}(\mathcal{P}))\nonumber\\
&\quad+\Es{T_{\mat{X}^n}|\mat{X}^n\hspace{-.1cm}\notin\mathscr{T}_{\epsilon}^{(n)}(\mathcal{P})}\Pr(\mat{X}^n\hspace{-.1cm}\notin\mathscr{T}_{\epsilon}^{(n)}(\mathcal{P}))-P^{(n)}_X\rVert_{\text{\tiny TV}}\nonumber\\
&\le\tvs{\Es{T_{\mat{X}^n}|\mat{X}^n\in\mathscr{T}_{\epsilon}^{(n)}(\mathcal{P})}\Pr(\mat{X}^n\in\mathscr{T}_{\epsilon}^{(n)}(\mathcal{P}))-P^{(n)}_X}\nonumber\\
&\quad+\tvs{\Es{T_{\mat{X}^n}|\mat{X}^n\notin\mathscr{T}_{\epsilon}^{(n)}(\mathcal{P})}\Pr(\mat{X}^n\notin\mathscr{T}_{\epsilon}^{(n)}(\mathcal{P}))}\label{eq:bound_input_etype}
\end{align}
for arbitrary $\epsilon>0$. 
Note that, for any two sequences $\mat{x}^n$ and $\mat{\tilde{x}}^n$ that belong to the set $\mathscr{T}_{\epsilon}^{(n)}(\mathcal{P})$, the convex combination of their types $T_{\mat{x}^n}$ and $T_{\mat{\tilde{x}}^n}$ satisfies
\begin{align*}
\tvs{\lambda T_{\mat{x}^n}+(1-\lambda)T_{\mat{\tilde{x}}^n}-P_X}<\epsilon 
\end{align*}
for some $P_X\in\mathcal{P}$ and any $\lambda\in[0,1]$. Thus, since 
\begin{align*}
\Es{T_{\mat{X}^n}|\mat{X}^n\in\mathscr{T}_{\epsilon}^{(n)}(\mathcal{P})}\Pr(\mat{X}^n\in\mathscr{T}_{\epsilon}^{(n)}(\mathcal{P}))
\end{align*}
is a convex combination of types of sequences in $\mathscr{T}_{\epsilon}^{(n)}(\mathcal{P})$, we have that
\begin{align*}
\tvs{\Es{T_{\mat{X}^n}|\mat{X}^n\in\mathscr{T}_{\epsilon}^{(n)}(\mathcal{P})}\Pr(\mat{X}^n\in\mathscr{T}_{\epsilon}^{(n)}(\mathcal{P}))-P^{(n)}_X}<\epsilon
\end{align*}
for some $P^{(n)}_X\in\mathcal{P}$. Regarding the second term in \eqref{eq:bound_input_etype}, we see that
\begin{align*}
&\tvs{\Es{T_{\mat{X}^n}|\mat{X}^n\notin\mathscr{T}_{\epsilon}^{(n)}(\mathcal{P})}\Pr(\mat{X}^n\notin\mathscr{T}_{\epsilon}^{(n)}(\mathcal{P}))}\\
&\quad\quad
=\Pr(\mat{X}^n\notin\mathscr{T}_{\epsilon}^{(n)}(\mathcal{P}))\tvs{\Es{T_{\mat{X}^n}|\mat{X}^n\notin\mathscr{T}_{\epsilon}^{(n)}(\mathcal{P})}}\\
&\quad\quad
\le\Pr(\mat{X}^n\notin\mathscr{T}_{\epsilon}^{(n)}(\mathcal{P}))\\
&\quad\quad
<\epsilon,
\end{align*}
where the inequality is satisfied for sufficiently large $n$. Combining the two bounds, we see that
\begin{align*}
\tvs{\E{T_{\mat{X}^n}}-P^{(n)}_X}<2\epsilon.
\end{align*}
Finally, we complete the proof by letting $\epsilon\rightarrow0$.
\end{proof}
We note that it is also possible to prove the preceding lemma by using the techniques in \cite{SV97} (in particular, \cite[Theorem~4]{SV97}), adapted to our notion of convergence.

We are now ready to prove Theorem~\ref{th:capacity_interference_pmf}.

\begin{proof}[Proof of Theorem~\ref{th:capacity_interference_pmf}]
The achievability result follows easily from Shannon's coding theorem. For the converse result, consider a sequence of $(n,2^{nR})$-codes that achieve the rate-interference type pair $(R,G_Z)$. 
The sequence, together with the uniform distribution on the messages, induces the joint distribution
\begin{align}
\frac{1}{\abs{\mathcal{M}}}P_{\mat{X}^n|M}P_{\mat{Y}^n|\mat{X}^n}P_{\mat{Z}^n|\mat{X}^n}P_{\hat{M}|\mat{Y}^n},\label{eq:convere_induced_dist}
\end{align}
with $P_{\mat{Y}^n|\mat{X}^n}=\prod P_{Y|X}$ and $P_{\mat{Z}^n|\mat{X}^n}=\prod P_{Z|X}$. Observe that in \eqref{eq:convere_induced_dist}, we have restricted our attention to distributions $P_{Y,Z|X}=P_{Y|X}P_{Z|X}$. As discussed before, this entails no loss of generality.

First, by the standard arguments based on Fano's inequality (e.g., see \cite[eq.~(3.3)]{GK11}), a vanishing error probability (i.e., \eqref{eq:achiev_cond1}) implies that
\begin{align}
nR
&\le\sum_{q=1}^n I(X_q;Y_q)+n\epsilon_n\nonumber\\
&=n\sum_{q=1}^n \frac{1}{n}I(X_q;Y_q|Q=q)+n\epsilon_n\nonumber\\
&= n I(X_Q;Y_Q|Q)+n\epsilon_n\nonumber\\
&\le n I(QX_Q;Y_Q)+n\epsilon_n\nonumber\\
&= n I(X_Q;Y_Q)+n\epsilon_n\label{eq:converse_last}
\end{align}
where $Q$ is a random variable uniformly distributed on $\{1,\ldots,n\}$ and independent of $(\mat{X}^n,\mat{Y}^n,\mat{Z}^n)$, and $\epsilon_n\ge0$ with $\epsilon_n\rightarrow0$ as $n\rightarrow\infty$. The last equality in \eqref{eq:converse_last} is justified by the fact that the DMC establishes the Markov chain $Q-X_Q-Y_Q$. Dividing by $n$, we obtain
\begin{align*}
R&\le I(X_Q;Y_Q)+\epsilon_n.
\end{align*}
This mutual information is evaluated for $P_{X_Q,Y_Q}$, which can be written as 
\begin{align*}
P_{X_Q,Y_Q}(x,y)&=P_{X_Q}(x)P_{Y|X}(y|x)\\
&=\E{T_{\mat{X}^n}(x)}P_{Y|X}(y|x).
\end{align*}
The first equality comes from the Markov chain $Q-X_Q-Y_Q$. The second equality is Property~2 in \cite[Section~VII.B.2]{CPC10}.

Now, condition \eqref{eq:achiev_cond2} on the type  of the interference for a sequence of $(n,2^{nR})$-codes that achieves the pair $(R,G_Z)$, combined with Lemma~\ref{lemma:conv_output_input}, implies that the expectation of the type of the input to the channel $\E{T_{\mat{X}^n}}$ must converge to a sequence $P_X^{(n)}$ with $P_X^{(n)}\in\mathcal{P}$ for all $n$. That is,
\begin{align*}
\E{T_{\mat{X}^n}(x)}P_{Y|X}(y|x)\rightarrow P_X^{(n)}(x)P_{Y|X}(y|x)
\end{align*}
or, equivalently, 
\begin{align*}
P_{X_Q,Y_Q}(x,y)\rightarrow P_X^{(n)}(x)P_{Y|X}(y|x).
\end{align*}
Since the mutual information is a continuous function of the input distribution, this convergence implies that any sequence of $(n,2^{nR})$-codes must satisfy
\begin{align*}
R&\le \limsup_{n\rightarrow\infty} I(X;Y)\rvert_{P_X^{(n)}}\\
&\le \max_{P_X\in\mathcal{P}} I(X;Y).
\end{align*}
In conclusion, achievability of the pair $(R,G_Z)$ implies that $(R,G_Z)\in\mathcal{C}$.
\end{proof}

\section{Multiple Users}
\label{sec:multiple_user}
Consider the scenario depicted in Figure~\ref{fig:TwoUserScenario}. Two transmitters want to communicate with their respective receivers through a channel governed by a conditional product pmf 
\begin{align}
P_{Y_1,Y_2,Z|X_1,X_2}=P_{Y_1|X_1}P_{Y_2|X_2}P_{Z|X_1,X_2}.\label{eq:multiple_pmf}
\end{align}
The marginals $P_{Y_1|X_1}$ and $P_{Y_2|X_2}$ model orthogonal communication channels between pairs of encoders and decoders, whereas $P_{Z|X_1,X_2}$ models the joint disturbance that the two transmissions create to the observer. That is, although the user pairs do not hamper each other's transmission, they create interference at a third external node, the observer. To control this interference, the two transmitters have access to a unidirectional rate-limited noiseless link from the first to the second encoder. They can use this resource to coordinate their transmissions and shape the type of the interference $T_{\mat{z}^n}(z)$.

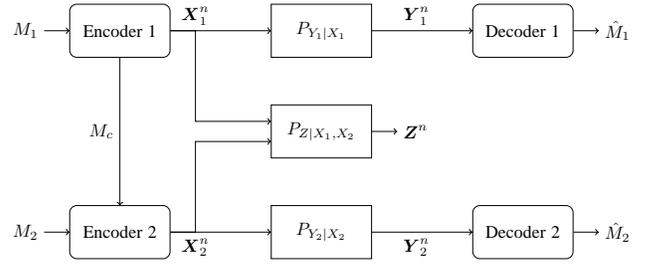
\begin{figure}[t]
\centering
\resizebox{\columnwidth}{!}{
\tikzstyle{line} = [draw, ->]
\tikzstyle{block} = [rectangle, draw, text width=5em, text centered, rounded corners, minimum height=3em]
\tikzstyle{channel} = [rectangle, draw, text width=5em, text centered, minimum height=3em]
\tikzstyle{invisible} = [draw=none]
\begin{tikzpicture}[node distance = 4cm, auto]
    \node [block] (Encoder_1) {Encoder 1};
	\node [channel, right of=Encoder_1] (ch_1) {$P_{Y_1|X_1}$};
	\node [channel, below of=ch_1, node distance = 2cm] (ch_3) {$P_{Z|X_1,X_2}$};
	\node [channel, below of=ch_3, node distance = 2cm] (ch_2) {$P_{Y_2|X_2}$};               	
    \node [block, right of=ch_1] (Decoder_1) {Decoder 1};
	\node [block, right of=ch_2] (Decoder_2) {Decoder 2};     
    \node [block, left of=ch_2] (Encoder_2) {Encoder 2};    
       
    \path [line] (Encoder_1) -- (ch_1) ;
    \path [line] (ch_1) -- ([xshift=.5cm]ch_1.east) [above right] node{$\mat{Y}_1^n$} -- (Decoder_1);    
    \path [line] (Encoder_2) -- (ch_2) ;
    \path [line] (ch_2) -- ([xshift=.5cm]ch_2.east) [below right] node{$\mat{Y}_2^n$} -- (Decoder_2); 
	\path [line] (Encoder_1.east) -- (1.5,0) [above] node{$\mat{X}_1^n$} -- (1.5,-1.8) -- (1.5,-1.8) -- ([yshift=.2 cm]ch_3.west); 
	\path [line] (Encoder_2.east) -- (1.5,-4) [below] node{$\mat{X}_2^n$} -- (1.5,-2.2) -- (1.5,-2.2) -- ([yshift=-.2 cm]ch_3.west); 
	\path [line] (Encoder_1) -- (0,-2) [left] node{$M_c$} -- (Encoder_2) ;
	\path [line] ([xshift=-.5cm]Encoder_1.west) [left] node{$M_1$} -- (Encoder_1.west);
	\path [line] (Decoder_1.east) -- ([xshift=.5cm]Decoder_1.east) [right] node{$\hat{M}_1$};	
	\path [line] ([xshift=-.5cm]Encoder_2.west) [left] node{$M_2$} -- (Encoder_2.west) ;	
	\path [line] (Decoder_2.east) -- ([xshift=.5cm]Decoder_2.east) [right] node{$\hat{M}_2$};	
	\path [line] (ch_3.east) -- ([xshift=.5cm]ch_3.east) [right] node{$\mat{Z}^n$};	
\end{tikzpicture}
}
\caption{Scenario for coordination of communications with interference constraints.}
\label{fig:TwoUserScenario}
\end{figure}

Observe that our model makes no assumption on how the two transmitters interfere with the observer, beyond the structure in \eqref{eq:multiple_pmf} (i.e., memoryless interference at symbol level). By choosing appropriately $P_{Z|X_1,X_2}$, we can model a scenarios ranging from symbol-level synchronization to carrier level synchronization, among others.

We now introduce the necessary definitions and state our main results for this scenario.

\begin{definition}[Code]
An $(n,2^{nR_1},2^{nR_2},2^{nR_c})$-code for the scenario in Figure~\ref{fig:TwoUserScenario} consists of: 
\begin{itemize}
\item three sets of messages: 
\begin{align*}
\mathcal{M}_j&\triangleq\{1,\ldots,\lceil 2^{nR_j}\rceil\}\text{ for }j\in\{1,2\},\\
\mathcal{M}_c&\triangleq\{1,\ldots,\lfloor 2^{nR_c}\rfloor\},
\end{align*}
\item two encoding functions
\begin{align*}
\mat{x}_1^n&:\mathcal{M}_1\rightarrow \mathcal{X}_1^n,\\
\mat{x}_2^n&:\mathcal{M}_2\times\mathcal{M}_c\rightarrow \mathcal{X}_2^n,
\end{align*}
\item a coordination function $c:\mathcal{M}_1\rightarrow \mathcal{M}_c$,
\item and two decoding functions $\hat{m}_j:\mathcal{Y}_j^n\rightarrow \mathcal{M}_j\cup\{e\}$ for $j\in\{1,2\}$.
\end{itemize}
\end{definition}
We assume that the message pair $(M_1,M_2)$ is uniformly distributed over the set $\mathcal{M}_1\times\mathcal{M}_2$. The notion of achievability and the definition of the communication-interference capacity region $\mathcal{C}$ are straightforward extensions of those introduced in the single user case. As for that case, the communication-interference capacity region $\mathcal{C}$ is convex. However, observe that the factorization in \eqref{eq:multiple_pmf} entails a loss of generality. 

Consider the following set:
\begin{align*}
\mathcal{R}\triangleq\left\{
\begin{array}{l}
(R_1,R_2,R_c,Q_Z)\text{ s.t. }\exists~P_UP_{X_1|U}P_{X_2|U}\text{ s.t. }\\
\quad\quad R_1<I(X_1;Y_1),\\
\quad\quad R_2<[I(X_2;Y_2)-I(U;X_2)]^+,\\
\quad\quad R_c> I(U;X_1),\\
\quad\quad\sum\limits_{u,x_1,x_2}P_UP_{X_1|U}P_{X_2|U}P_{Z|X_1,X_2}=Q_{Z}
\end{array}
\right\}
\end{align*}
where $[x]^+\triangleq\max(x,0)$. Let $\operatorname{conv}(\mathcal{R})$ denote the convex hull of $\mathcal{R}$. Our main result for the channel model in Figure~\ref{fig:TwoUserScenario} is the following partial characterization.

\begin{theorem}\label{th:interfcontr_mu}
The communication-interference capacity region $\mathcal{C}$ satisfies
\begin{align*}
\operatorname{conv}(\mathcal{R})\subseteq\mathcal{C}.
\end{align*}
\end{theorem}

Before proving the theorem, we make the following two observations about $\mathcal{R}$: i) The random variable $U$ plays the role of the coordination message sent from Encoder~1 to Encoder~2. By setting $U=\emptyset$, we obtain $R_c=0$ and recover the case where the users are not coordinated (i.e., $X_1$ and $X_2$ are independent). For most distributions $P_{Z|X_1,X_2}$, our strategy strictly improves upon uncoordinated communication. ii) The coordination message $U$ couples the rates $R_1$ and $R_2$ in two ways. First, the choices of input distributions have to be compatible in the sense that they yield the desired $G_Z$. In addition, the rate for Encoder~2 has a penalty term that reflects that the transmitted signals are correlated. That is, $X_2$ carries information about $X_1$. This is similar to the situation in Gel'fand Pinsker coding, where the transmission is aligned with the channel state and thus carries information about it \cite{GP80}. These considerations are illustrated by the following example.

\begin{example}\label{example:constellation}
Consider the scenario in which each of the two encoders can make use of the set of $16$ symbols depicted in Figure~\ref{fig:constellation16} as inputs to the channel. Assume that the observer tolerates only low and mild levels of interference. This means that the two encoders are not allowed to use the black-circle symbols simultaneously. For simplicity, assume that the channels $P_{Y_1|X_1}$ and $P_{Y_2|X_2}$ are noiseless.

Without coordination, one of the two users is restricted to use only the subset of red-diamond symbols. Assume that the restriction is placed on the second user. This yields the rate pair $(R_1,R_2)=(4,2)$. In contrast, if Encoder~$1$ uses the coordination link to declare whether it will use a black-circle or a red-diamond symbol, Encoder~$2$ can opportunistically choose its constellation to boost its communication rate. For example, if Encoder~$1$ makes use of all $16$ symbols with equal frequency, then Encoder~$2$ is forced to use the red-diamond symbols (i.e., transmit $2$ [bpcu]) $75\%$ of the times. However, in the remaining $25\%$, it can use any of the black-circle symbols (i.e., $\log_2 12$ [bpcu]). This yields
\begin{align*}
R_2=\frac{3}{4}2+\frac{1}{4}\log_2 12\approx 2.4\text{~[bpcu]}.
\end{align*}
Thus, we have $(R_1,R_2)=(4,2.4)$. Observe that the constraint placed by the observer does not preclude Encoder~$2$ from using \emph{any} of the symbols in Figure~\ref{fig:constellation16} when Encoder~$1$ sends a red-diamond symbol. However, Decoder~$2$ needs to know whether the transmitted symbol corresponds to $2$ or $4$ bits. By restricting its input to belong to the set of black-circle symbols, Encoder~$2$ is conveying information about the message of Encoder~$1$, namely that the current input consists of one of the red-diamond symbols.

\begin{figure}[t]
\centering
\resizebox{0.4\columnwidth}{!}{
\tikzstyle{arrow} = [draw, ->]
\tikzstyle{invisible} = [draw=none, scale=0.5]
\tikzstyle{dot} = [circle, draw, fill, scale=0.75]
\tikzstyle{dotr} = [diamond, draw, fill, scale=0.75, color=red!60]
\begin{tikzpicture}[node distance = 2cm, auto]
    \node [dot] (S11){};
    \node [dot, below of=S11] (S21){};
    \node [dot, below of=S21] (S31){};
    \node [dot, below of=S31] (S41){};

    \node [dot, right of=S11] (S12){};
    \node [dotr, below of=S12] (S22){};
    \node [dotr, below of=S22] (S32){};
    \node [dot, below of=S32] (S42){}; 

    \node [dot, right of=S12] (S13){};
    \node [dotr, below of=S13] (S23){};
    \node [dotr, below of=S23] (S33){};
    \node [dot, below of=S33] (S43){};   
    
    \node [dot, right of=S13] (S14){};
    \node [dot, below of=S14] (S24){};
    \node [dot, below of=S24] (S34){};
    \node [dot, below of=S34] (S44){};                    
\end{tikzpicture}
}
\caption{Constellation with $16$ symbols in Example~\ref{example:constellation}. The constraint on the interference at the observer precludes transmission of black-circle symbols by both encoders at the same time.}
\label{fig:constellation16}
\end{figure}
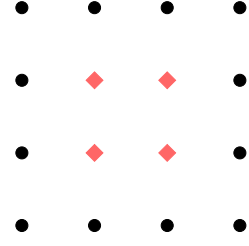

A coordination rate equal to $R_c=0.81$ [bpa] is sufficient to implement this protocol if Encoder~$1$ uses a lossless source coding algorithm to declare its intentions for a batch of channel uses.
\end{example}

\begin{proof}[Proof of Theorem~\ref{th:interfcontr_mu}]
Fix arbitrary $\epsilon>0$ and let $\delta(\epsilon)>0$ be some positive function such that $\delta(\epsilon)\rightarrow0$ as $\epsilon\rightarrow0$. Choose a tuple $(R_1,R_2,R_c,Q_Z)\in\mathcal{R}$ and let $\tilde{R}_2>R_2$. Let $P_UP_{X_1|U}P_{X_2|U}$ be the corresponding distribution.

\subsection*{Codebook generation}
\begin{itemize}
\item For every $m_c\in\mathcal{M}_c$, generate a sequence $\mat{u}^n(m_c)$ according to $\prod_{i=1}^nP_U(u_i)$.
\item For every $m_1\in\mathcal{M}_1$, generate a codeword $\mat{x}_1^n(m_1)$ according to $\prod_{i=1}^nP_{X_1}(x_{1i})$.
\item For every $m_2\in\mathcal{M}_2$ and every $l\in\{1,..,\lceil 2^{n(\tilde{R}_2-R_2)}\rceil\}$, generate a codeword $\mat{x}_2^n(l,m_2)$ according to $\prod_{i=1}^nP_{X_2}(x_{2i})$.
\end{itemize}

\subsection*{Encoding}
\begin{enumerate}
\item To transmit the message $m_1$, Encoder~$1$ puts the codeword $\mat{x}_1^n(m_1)$ into the channel.
\item To generate the coordination message given $\mat{x}_1^n(m_1)$, Encoder~$1$ searches for an index $m_c$ such that $(\mat{u}^n(m_c),\mat{x}_1^n(m_1))\in\mathcal{T}_{\epsilon}^{(n)}(P_{U,X_1})$. If more than one such $m_c$ exists, it chooses one at random among the candidates. If none exists, then it chooses $m_c=1$. Finally, it conveys the index $m_c$ to Encoder~$2$.
\item To transmit the message $m_2$, Encoder~$2$ searches for an index $l$ such that $(\mat{u}^n(m_c),\mat{x}_2^n(l,m_2))\in\mathcal{T}_\epsilon^{(n)}(P_{U,X_2})$. If more than one such $l$ exists, it chooses one at random among the candidates. If none exists, then it chooses $l=1$. Finally, it puts the codeword $\mat{x}_2^n(l,m_c)$ into the channel. 
\end{enumerate}

\subsection*{Decoding}
\begin{itemize}
\item Given the observation $\mat{y}_1^n$, Decoder~$1$ searches for a unique index $\hat{m}_1$ such that 
$(\mat{x}_1^n(\hat{m}_1),\mat{y}_1^n)\in\mathcal{T}_{\epsilon}^{(n)}(P_{X_1,Y_1})$. If no such $\hat{m}_1$ is found or if it is not unique, the decoder declares an error. 
\item Given the observation $\mat{y}_2^n$, Decoder~$2$ searches for a unique index $\hat{m}_2$ such that 
$(\mat{x}_2^n(\hat{l},\hat{m}_2),\mat{y}_2^n)\in\mathcal{T}_{\epsilon}^{(n)}(P_{X_2,Y_2})$ for some $\hat{l}\in\{1,\ldots,\lceil2^{n(\tilde{R}_2-R_2)}\rceil\}$. If no such $\hat{m}_2$ is found or if it is not unique, the decoder declares an error. 
\end{itemize}

\subsection*{Analysis of the error probability}
We consider the error probability averaged over the ensemble of codebooks. Let $\mathcal{E}$ denote the error event and consider a fixed $n$. Due to the symmetry in the generation of the codebooks, we can assume that $M_1=M_2=1$ without loss of generality. That is,
\begin{align*}
\Pr(\mathcal{E})=\Pr(\mathcal{E}|(M_1,M_2)=(1,1)).
\end{align*}

To bound the error probability, consider the following events:
\begin{align*}
\mathcal{E}_Z&\triangleq\{\tvs{T_{\mat{Z}^n}-Q_Z}\ge\epsilon\},\\
\mathcal{E}_i&\triangleq\{\hat{M}_i\ne 1\}
\end{align*}
for $i=\{1,2\}$. The error probability satisfies 
\begin{align}
\Pr(\mathcal{E})
&\le\Pr(\mathcal{E}_Z|(M_1,M_2)=(1,1))\nonumber\\
&\quad+\Pr(\mathcal{E}_1|M_1=1)+\Pr(\mathcal{E}_2|M_2=1).\label{eq:3bounds}
\end{align}
We bound each of the three terms individually. For the first term in \eqref{eq:3bounds}, consider the event
\begin{align*}
\mathcal{E}_{Z0}\hspace{-.1cm}
\triangleq\hspace{-.1cm}\{\hspace{-.05cm}\tvs{T_{\mat{U}^n\hspace{-.1cm},\mat{X}_1^n(1),\mat{X}_2^n(L,1),\mat{Z}^n}\hspace{-.1cm}-\hspace{-.1cm}P_{Z|X_1,X_2}\hspace{-.05cm}P_{X_1|U}\hspace{-.05cm}P_{X_2|U}\hspace{-.05cm}P_{U}\hspace{-.05cm}}\hspace{-.1cm}\ge\hspace{-.1cm}\epsilon\hspace{-.05cm}\}\\
\end{align*}
and note that, by the basic properties of strong typicality, for every $(\mat{u}^n,\mat{x}_1^n,\mat{x}_2^n,\mat{z}^n)$ such that
\begin{align*}
\tvs{T_{\mat{u}^n,\mat{x}_1^n,\mat{x}_2^n,\mat{z}^n}-P_{Z|X_1,X_2}P_{X_1|U}P_{X_2|U}P_{U}}<\epsilon,
\end{align*}
we have
\begin{align*}
\tvs{T_{\mat{z}^n}-Q_Z}<\epsilon.
\end{align*}
Therefore,
\begin{align*}
\Pr(\mathcal{E}_Z|(M_1,M_2)=(1,1))\le\Pr(\mathcal{E}_{Z0}).
\end{align*}
Now, let $\epsilon'=\frac{\epsilon}{4}$ and
\begin{align*}
\mathcal{E}_{Z1}&\triangleq\{(\mat{U}^n(m_c),\mat{X}_1^n(1))\hspace{-.05cm}\notin\hspace{-.05cm}\mathcal{T}_{\epsilon'}^{(n)}(P_{U,X_1})\text{ for all }m_c\hspace{-.05cm}\in\hspace{-.05cm}\mathcal{M}_c\},\\
\mathcal{E}_{Z2}&\triangleq\{(\mat{U}^n(M_c),\mat{X}_2^n(l,1))\notin\mathcal{T}_{\epsilon'}^{(n)}(P_{U,X_2})\nonumber\\
&\quad\quad\quad\text{ for all }l\in\{1,\ldots,\lceil 2^{n(\tilde{R}_2-R_2)}\rceil\}\},\\
\mathcal{E}_{Z3}&\triangleq\{(\mat{U}^n(M_c),\mat{X}_1^n(1),\mat{X}_2^n(L,1))\notin\mathcal{T}_{\epsilon}^{(n)}(P_{U,X_1,X_2})\},\\
\mathcal{E}_{Z4}&\triangleq\{(\mat{U}^n(M_c),\mat{X}_1^n(1),\mat{X}_2^n(L,1)\hspace{-.05cm},\hspace{-.05cm}\mat{Z}^n\hspace{-.05cm})\hspace{-.1cm}\notin\hspace{-.05cm}\mathcal{T}_{\epsilon}^{(n)}\hspace{-.05cm}(\hspace{-.05cm}P_{U,X_1,X_2,Z}\hspace{-.05cm})\hspace{-.05cm}\}\hspace{-.05cm}.
\end{align*}
Here $M_c$ and $L$ are the random variables corresponding to the coordination index and the index chosen by Encoder~2, respectively. We have that 
\begin{align}
\Pr(\mathcal{E}_{Z0})&\le\Pr(\mathcal{E}_{Z1})+\Pr(\mathcal{E}_{Z2})\nonumber\\
&\quad+\Pr(\mathcal{E}_{Z3}\cap(\mathcal{E}_{Z1}^c\cap\mathcal{E}_{Z2}^c))
+\Pr(\mathcal{E}_{Z4}\cap\mathcal{E}_{Z3}^c).\label{eq:Puxyz}
\end{align}
By the covering lemma \cite[Lemma~3.3]{GK11}, $\Pr(\mathcal{E}_{Z1})\rightarrow0$ as $n\rightarrow\infty$ if $R_c>I(U;X_1)-\delta(\epsilon')$. For the second term in \eqref{eq:Puxyz}, note that the distribution of $(\mat{U}^n(M_c),\mat{X}_2^n(l,1))$ is the same for all values of $M_c$ and $l$; they are independent. Thus, again by the covering lemma, $\Pr(\mathcal{E}_{Z2})\rightarrow0$ as $n\rightarrow\infty$ if $\tilde{R}_2-R_2>I(U;X_2)-\delta(\epsilon')$. 

Regarding the third term in \eqref{eq:Puxyz}, we observe the following. Given $\mathcal{E}_{Z1}^c$, we have that $(\mat{U}^n(M_c),\mat{X}_1^n(1))\in\mathcal{T}_{\epsilon}^{(n)}(P_{U,X_1})$. Similarly, given $\mathcal{E}_{Z2}^c$, we have that $(\mat{U}^n(M_c),\mat{X}_2^n(L,1))\in\mathcal{T}_{\epsilon}^{(n)}(P_{U,X_2})$. Thus, by the strong Markov Lemma \cite[Theorem~12]{CPC10}, $\Pr(\mathcal{E}_{Z3}\cap(\mathcal{E}_{Z1}^c\cap\mathcal{E}_{Z2}^c))\rightarrow0$ as $n\rightarrow\infty$. The conditions of the lemma are satisfied because $X_1-U-X_2$ form a Markov chain and the distribution of $\mat{X}_2^n$ is permutation invariant (as defined in \cite{CPC10}) with respect to $\mat{u}^n$.

Finally, for the last term in \eqref{eq:Puxyz}, we have that $\mat{Z}^n$ is generated by passing a $\epsilon$-typical pair $(\mat{X}_1^n,\mat{X}_2^n)$ through the channel $P_{Z|X_1,X_2}$. Thus, by the law of large numbers, $\Pr(\mathcal{E}_{Z4}\cap\mathcal{E}_{Z3}^c)\rightarrow0$ as $n\rightarrow0$.

We now turn our attention to the term $\Pr(\mathcal{E}_1|M_1=1)$ in \eqref{eq:3bounds}. Consider the following events
\begin{align*}
\mathcal{E}_{11}&\triangleq\{(\mat{X}_1^n(1),\mat{Y}_1^n)\notin\mathcal{T}_\epsilon^{(n)}(P_{X_1,Y_1})\},\\
\mathcal{E}_{12}&\triangleq\{(\mat{X}_1^n(\hat{m}_1),\mat{Y}_1^n)\in\mathcal{T}_\epsilon^{(n)}(P_{X_1,Y_1})\text{ for some }\hat{m}_1\ne1\}.
\end{align*}
We have that
\begin{align*}
\Pr(\mathcal{E}_1|M_1=1)\le \Pr(\mathcal{E}_{11})+\Pr(\mathcal{E}_{12}),
\end{align*}
where $\Pr(\mathcal{E}_{11})\rightarrow0$ as $n\rightarrow0$ by the law of large numbers, and $\Pr(\mathcal{E}_{12})\rightarrow0$ as $n\rightarrow0$ if 
$R_1<I(X_1;Y_1)-\delta(\epsilon)$ by the packing lemma \cite[Lemma~3.1]{GK11}.

Similarly, if $\tilde{R}_2<I(X_2;Y_2)-\delta(\epsilon)$ then $\Pr(\mathcal{E}_2|M_2=1)\rightarrow 0$ as $n\rightarrow0$. Combining all the terms and letting $\epsilon\rightarrow0$, we obtain
\begin{align*}
R_c&>I(U;X_1),\\
R_1&<I(X_1;Y_1),\\
R_2&<[\tilde{R}_2-I(U;X_2)]^+<[I(X_2;Y_2)-I(U;X_2)]^+,
\end{align*}
as desired. The remaining tuples in the convex hull are achieved by time sharing.
\end{proof}

\section{Conclusion}
\label{sec:concl}
We have proposed a generic model in terms of types (i.e., empirical distributions) for studying the effect of the interference induced by a communication process. First, we have considered the case of a single communication link and shown the existence of a tradeoff between the rate of communication and the type of the induced interference. 
To quantify this tradeoff, we have introduced the notion of communication-interference capacity region and we have explicitly characterized it. Then, we have studied a multiple-user scenario with unidirectional coordination of the transmitters. In this case, we have shown that the tradeoff involves the interference type and the communication rate as well as the coordination rate. We have established an inner bound to the communication-interference capacity region as a partial characterization of the tradeoff.


\bibliographystyle{IEEEtran}
\bibliography{reference}

\end{document}